\documentclass[oribibl]{llncs}

\title{Approximation Algorithms for Connected Maximum Cut and Related Problems}
\author{MohammadTaghi Hajiaghayi\inst{1}\thanks{Partially supported by NSF CAREER Award 1053605, NSF grant CCF-1161626, DARPA/AFOSR grant FA9550-12-1-0423, and a Google Faculty Research award.}
	\and Guy Kortsarz\inst{2}\thanks{Partially supported by NSF grant 1218620.}
        \and Robert MacDavid\inst{2}
        \and Manish Purohit\inst{1}\thanks{Partially supported by NSF grants CCF-1217890 and IIS-1451430.}
	\and Kanthi Sarpatwar\inst{3}\thanks{Partially supported by NSF grant CCF-1217890. Work done when the author was a student
    at the University of Maryland, College Park.}}

      \institute{University of Maryland, College Park, MD 20742, USA \\ {\tt
        	\{hajiagha, manishp\}@cs.umd.edu}  \and 
        	Rutgers University - Camden, Camden, NJ 08102, USA \\ {\tt
        guyk@camden.rutgers.edu, robertmacdavid@gmail.com}\and IBM T. J. Watson
  Research Center, Yorktown Heights, NY 10598\\ {\tt sarpatwa@us.ibm.com}}

\usepackage[lined,boxed,linesnumbered, commentsnumbered]{algorithm2e}
\usepackage{caption}
\usepackage{amsmath, amssymb}
\usepackage{subcaption}
\usepackage{subfig}
\usepackage{graphicx}
\usepackage[margin=1.5in]{geometry}

\newcommand{\slfrac}[2]{{#1/#2}}
\newcommand{\mone}{{\eta}}
\newcommand{\alphatt}{$\alpha$-\textsf{thick}}

\newcommand{\mc}{\textsf{Max-Cut}}
\newcommand{\cmcfull}{\textsf{Connected Maximum Cut}}
\newcommand{\cmc}{\textsf{CMC}}
\newcommand{\bcmc}{\textsf{b-CMC}}
\newcommand{\wcmc}{\textsf{WCMC}}
\newcommand{\pmsat}{\textsf{PM-3SAT}}

\newcommand{\opt}{S}

\usepackage{thmtools}
\usepackage{thm-restate}

\DeclareMathOperator*{\argmin}{arg\,min}
\DeclareMathOperator*{\argmax}{arg\,max}

\begin{document}

\maketitle

\begin{abstract}
An instance of the \cmcfull\ problem consists of an undirected graph $G=(V,E)$ and the goal is to find a subset of vertices $S \subseteq V$ that maximizes the number of edges in the cut $\delta(S)$ such that the induced graph $G[S]$ is connected. 
We present the first non-trivial $\Omega(\frac{1}{\log n})$ approximation algorithm for the \cmcfull\ problem in general graphs using novel techniques. We then extend our algorithm to edge weighted case and obtain a poly-logarithmic approximation algorithm.
Interestingly, in contrast to the classical Max-Cut problem that can be solved in polynomial time on planar graphs, we show that the \cmcfull\ problem remains NP-hard on unweighted, planar graphs. On the positive side, we obtain a polynomial time approximation scheme for the \cmcfull\ problem on planar graphs and more generally on bounded genus graphs.
\end{abstract}

\section{Introduction}
\label{sec:introduction}

Submodular optimization problems have, in recent years, received a considerable amount of attention
\cite{badanidiyuru2014fast, buchbinder2012tight, buchbinder2014submodular, calinescu2007maximizing,chekuri2011submodular, feige2011maximizing, nemhauser1978analysis}
in algorithmic research. 
In a general \textsf{Submodular Maximization} problem, we are given a non-negative submodular\footnote{A function $f$ is called submodular if $f(S) + f(T) \geq f(S \cup T) + f(S \cap T)$ for all $S, T \subseteq U$.} function over the power set of a universe $U$ of elements, $f:2^U\rightarrow \mathbb{R}^+\cup \{0\}$ and the goal is to find a subset $S\subseteq U$ that maximizes $f(S)$ so that $S$ satisfies certain pre-specified constraints. 
In addition to their practical relevance, the study of submodular maximization problems has led to the development of several important theoretical techniques such as the continuous greedy method and multi-linear extensions~\cite{calinescu2007maximizing} and the double greedy~\cite{buchbinder2012tight} algorithm, among others.

In this study, we are interested in the problem of maximizing a submodular set function over vertices of a graph, such that the selected vertices induce a connected subgraph.
Motivated by applications in coverage over wireless networks,
Kuo et al.~\cite{kuo2013maximizing} consider the problem of maximizing a monotone, submodular function $f$ subject to connectivity and cardinality constraints of the form $|S|\leq k$ and provide an $\Omega(\frac{1}{\sqrt{k}})$ approximation algorithm. 
For a restricted class of monotone, submodular functions that includes the covering function\footnote{In this context, a covering function is defined as  $f(S) = \sum_{v \in \mathbb{N}^+(S)} weight(v)$ where $\mathbb{N}^+(S)$ is the closed neighborhood of the set of vertices $S$}, Khuller et al.~\cite{khuller2014analyzing} give a constant factor approximation to the problem of maximizing $f$ subject to connectivity and cardinality constraints.

In the light of these results, it is rather surprising that no non-trivial approximation algorithms are known for the case of general (non-monotone) submodular functions. 
Formally, we are interested in the following problem, which we refer to as 
\textsf{Connected Submodular Maximization (CSM)}:  Given a simple, undirected graph $G=(V,E)$ and a non-negative submodular set function $f:2^V\rightarrow \mathbb{R}^+\cup\{0\}$, find a subset of vertices $S\subseteq V$ that maximizes $f(S)$ such that $G[S]$ is connected. 
We take the first but important step in this direction and study the problem in the case of one of the most important non-monotone submodular functions, namely the \textsf{Cut} function.
Formally, given an undirected graph $G=(V,E)$, the goal is to find a subset $S \subseteq V$, such that $G[S]$ is connected and the number of edges that have exactly one end point in $S$, referred to as the cut function $\delta(S)$, is maximized. We refer to this as the \cmcfull\ problem.
Further, we also consider an edge weighted variant of this problem, called the \textsf{Weighted }\cmcfull\ problem, where function to be maximized is the total weight of edges in the cut $\delta(S)$. 

We now outline an application to the image segmentation problem that seeks to identify ``objects'' in an image. Graph based approaches for image segmentation \cite{felzenszwalb2004efficient, petrov2005image} represent each pixel as a vertex and weighted edges represent the dissimilarity (or similarity depending on the application) between adjacent pixels. Given such a graph, a connected set of pixels with a large weighted cut naturally corresponds to an object in the image. Vicente et al.~\cite{vicente2008graph} show that even for interactive image segmentation, techniques that require connectivity also perform significantly better that cut based methods alone.

\subsection*{Related Work}
\mc\ is a fundamental problem in combinatorial optimization that finds applications in diverse areas. 
 A simple randomized algorithm that adds each vertex to $S$ independently with probability $\slfrac{1}{2}$ gives a $0.5$-approximate solution in expectation.
In a breakthrough result, Goemans and Williamson \cite{goemans1995improved} gave a $0.878$-approximation algorithm using semidefinite programming and randomized rounding. 
Further, Khot et al.~\cite{khot2007optimal} showed that this factor is optimal assuming the Unique Games Conjecture. Interestingly, the \mc\ problem can be optimally solved in polynomial time in planar graphs  by a curious connection to the matching problem in the dual graph~\cite{hadlock1975finding}.
To the best of our knowledge, the \cmcfull\ problem has not been considered before our work. 
Haglin and Venkatesan~\cite{haglin1991approximation} showed that a related problem, where we require both sides of the cut, namely $S$ and $V \setminus S$, to be connected, is NP-hard in planar graphs.

We note that the well studied \textsf{Maximum Leaf Spanning Tree} (MLST) problem (e.g. see \cite{solis19982}) is a special case of the \textsf{Connected Submodular Maximization} problem. 
We also note that recent work on graph connectivity under vertex sampling leads to a simple constant approximation to the \textsf{Connected Submodular Maximization} for highly connected graphs, i.e., for graphs with $\Omega(\log n)$ vertex connectivity. Proofs of these claims are presented in the Appendix~\ref{app:mlst} and \ref{app:conn} respectively. 
 
We conclude this section by noting that connected variants of many classical combinatorial problems have been extensively studied in the literature and have been found to be useful. 
The best example for this is the \textsf{Connected Dominating Set} problem. Following the seminal work of Guha and Khuller~\cite{guha1998approximation}, the problem has found extensive applications (with more than a thousand citations) in the domain of wireless ad hoc networks as a \emph{virtual backbone} (e.g. see~\cite{das1997routing, du2013connected}). Few other examples of connected variants of classic optimization problems include \textsf{Group Steiner Tree}~\cite{garg1998polylogarithmic} (which can be seen as a generalization of a connected variant of \textsf{Set Cover}), \textsf{Connected Domatic Partition}~\cite{censor2015tight, censor2014new}, \textsf{Connected Facility Location}~\cite{eisenbrand2008approximating,swamy2004primal}, and \textsf{Connected Vertex Cover}~\cite{cygan2012deterministic}.

\subsection*{Contribution and Techniques}
\label{sec:our-contribution}
Our key results can be summarized as follows.

1. We obtain the first $\Omega(\frac{1}{\log n})$ approximation algorithm for the \cmcfull\ (\cmc) problem in general graphs. Often, for basic connectivity problems on graphs, one can obtain simple $O(\log n)$ approximation algorithms using a probabilistic embedding into trees with $O(\log n)$ stretch~\cite{fakcharoenphol2003tight}. Similarly, using the cut-based decompositions given by R\"{a}cke~\cite{racke2008optimal}, one can obtain $O(\log n)$ approximation algorithms for cut problems (e.g. Minimum Bisection). Interestingly, since the \cmc\ problem has the flavors of both cut and connectivity problems simultaneously, neither of these approaches are applicable. Our novel approach is to look for \textsf{$\alpha$-thick trees}, which are basically sub-trees with ``high'' degree sum on the leaves. 

2. For the \textsf{Weighted }\cmcfull\ problem, we obtain an $\Omega(\frac{1}{\log^2n})$ approximation algorithm. The basic idea is to group the edges into logarithmic number of weight classes and show that the problem on each weight class boils down to the special case where the weight of every edge is either $0$ or $1$.

3. We obtain a polynomial time approximation scheme for the 
\cmc\ problem in planar graphs and more generally in bounded genus graphs. This requires the application of a stronger form of the edge contraction theorem by Demaine, Hajiaghayi and Kawarabayashi~\cite{demaine2011contraction} that may be of independent interest.

4. We show that the \cmc\ problem remains NP-hard even on unweighted, planar graphs.  This is in stark contrast with the regular \mc\ problem that can be solved optimally in planar graphs in polynomial time. We obtain a polynomial time reduction from a special case of \textsf{3-SAT} called the \textsf{Planar Monotone 3-SAT} (\pmsat), to the \cmc\ problem in planar graphs. This entails a delicate construction, exploiting  the so called  ``rectilinear representation'' of a \pmsat\ instance, to maintain planarity of the resulting \cmc\ instance. 

\section{Approximation Algorithms for General Graphs}
\label{sec:appr-algor-gener}

In this section, we consider the \cmcfull\ problem in general graphs. In fact, we provide an $\Omega(\frac{1}{\log n})$ approximation algorithm for the more general problem in which edges can have weight 0 or 1 and the objective is to maximize the number of edges of weight 1 in the cut. This generalization will be useful later in obtaining a poly-logarithmic approximation algorithm for arbitrary weighted graphs.

We denote the cut of a subset of vertices $S$ in a graph $G$, i.e., the set of edges in $G$ that are incident on exactly one vertex of $S$ by $\delta_G(S)$ or when $G$ is clear from context, just $\delta(S)$. Further, for two disjoint subsets of vertices $S_1$ and $S_2$ in $G$, we denote the set of edges that have one end point in each of $S_1$ and $S_2$, by $\delta_G(S_1, S_2)$ or simply $\delta(S_1,S_2)$. 
The formal problem definition follows - 

\vspace{2mm}
\noindent
{\bf Problem Definition. }\textsf{\{0,1\}-Connected Maximum Cut} (\bcmc): Given a graph $G = (V,E)$ and a weight function $w:E \rightarrow \{0,1\}$, find a set $S \subset V$ that maximizes $\sum_{e \in \delta(S)} w(e)$ such that $G[S]$ induces a connected subgraph.
\vspace{2mm}

We call an edge of weight $0$, a \textsf{0-edge} and that of weight $1$, a \textsf{1-edge}. Further, let $w(\delta(S)) = \sum_{e \in \delta(S)} w(e)$ denote the weight of the cut, i.e., the number of \textsf{1-edges} in the cut. 
We first start with a simple reduction rule that ensures that every vertex $v \in V$ has at least one \textsf{1-edge} incident on it.

\begin{restatable}{clm}{onlygood}
Given a graph $G = (V,E)$, we can construct a graph $G' = (V',E')$ in polynomial time, such that every $v' \in V'$ has at least one \textsf{1-edge} incident on it and $G'$ has a \bcmc\ solution $S'$ of weight 
 at least $\psi$ if and only if $G$ has a \bcmc\ solution $S$ of weight at least $\psi$.
\label{lem:only-good}
\end{restatable}

\begin{proof}
  Let $v \in V$ be a vertex in $G$ that has only \textsf{0-edges} incident on it and let $\{v_1,v_2,\ldots,v_l\}$ denote the set of its neighbors. Consider the graph $G'$ obtained from $G$ by deleting $v$ along with all its incident edges and adding \textsf{0-edges} between every pair of its neighbors $\{v_i,v_j\}$ such that $\{v_i, v_j\} \notin E$. Let $S$ denote a feasible solution of weight $\psi$ in $G$. If $v \notin S$, then clearly $S' = S$ is the required solution in $G'$. If $v \in S$, we set $S' = S \setminus \{v\}$ and we claim that $G'[S']$ is connected if $G[S]$ is connected and $\sum_{e \in \delta_{G'}(S')} w(e) = \sum_{e \in \delta_G(S)} w(e)$. The latter part of the claim is true since all the edges that we delete and add are \textsf{0-edges}. To prove the former part, notice that if $v$ is not a cut vertex in $G[S]$ then $G[S']$ must be connected. On the other hand, even if $v$ is a cut vertex, the new edges added among all pairs of $v$'s neighbors ensure that $G[S']$ is connected. Finally, to prove the other direction, suppose we have a feasible solution $S'$ of weight $\psi$ in $G'$. Now, if $G[S']$ is connected, then $S = S'$ is a feasible solution in $G$ of weight $\psi$. Otherwise, set $S = S' \cup \{v\}$. Since $v$ creates a path between all pairs of its neighbors, $G[S]$ is connected if $G'[S']$ is connected and is thus a feasible solution of the same weight. The proof of the lemma follows from induction.
\qed
\end{proof}

From now on, we will assume, without loss of generality, that every vertex of $G$ has at least one \textsf{1-edge} incident on it. 
We now introduce some new definitions that would help us to present the main algorithmic ideas.
We denote by $W_G(v)$ the total weight of edges incident on a vertex $v$ in $G$, i.e., $W_G(v) = \sum_{e:v\in e} w(e)$. In other words, $W_G(v)$ is total number of \textsf{1-edges} incident on $v$. Further let $\mone$ be the total number of \textsf{1-edges} in the graph.
The following notion of an $\alpha$-thick tree is a crucial component of our algorithm.

 \begin{definition}[$\alpha$-Thick Tree]
Let $G=(V,E)$ be a graph with $n$ vertices and $\mone$ \textsf{1-edges}. A subtree $T \subseteq G$ (not necessarily spanning), with leaf set $L$, is said to be $\alpha$-thick if $\sum_{v \in L} W_G(v) \geq \alpha \mone$.
 \end{definition}

The following lemma shows that this notion of an $\alpha$-thick tree is intimately connected with the \bcmc\ problem.
\begin{restatable}{lem}{alphaalpha}
\label{lem:alphatt}
  For any $\alpha > 0$, given a polynomial time algorithm $A$ that computes an $\alpha$-thick tree $T$ of a graph $G$, we can obtain an $\frac{\alpha}{4}$-approximation algorithm  for the \bcmc\ problem on $G$.
\end{restatable}

\begin{proof}
Given a graph $G = (V,E)$ and weight function $w:E\rightarrow \{0,1\}$, we use Algorithm $A$ to compute an $\alpha$-thick tree $T$, with leaf set $L$. 
Let $m_L$ denote the number of \textsf{1-edges} in $G[L]$, the  subgraph  induced by $L$ in the graph $G$. We now partition $L$ into two disjoint sets $L_1$ and $L_2$ such that the number of \textsf{1-edges} in $\delta(L_1,L_2) \geq \frac{m_L}{2}$. 
This can be done by applying the standard randomized algorithm for \mc\ (e.g. see ~\cite{motwani1995randomized}) on $G[L]$ after deleting all the \textsf{0-edges}.   
Now, consider the two connected subgraphs $T \setminus L_1$ and $T \setminus L_2$. We first claim that every \textsf{1-edge} in $\delta(L)$ belongs to either $\delta(T \setminus L_1)$ or $\delta(T \setminus L_2)$. Indeed, any \textsf{1-edge} $e$ in $\delta(L)$, belongs to one of the four possible sets, namely  $\delta(L_2, T\setminus L)$, $\delta(L_1, V\setminus T)$, $\delta(L_1, T\setminus L)$ and $\delta(L_2, V\setminus T)$. In the first two cases, $e$ belongs to $\delta(T\setminus L_2)$ while in the last two cases, $e$ belongs $\delta(T\setminus L_1)$, hence the claim. Further, every \textsf{1-edge} in $\delta(L_1, L_2)$ belongs to both $\delta(T \setminus L_1)$ and $\delta(T \setminus L_2)$. Hence, we have - 
\begin{align}
  \label{eq:1}
  \sum_{e\in \delta(T \setminus L_1)} w(e) +  \sum_{e\in \delta(T \setminus L_2)} w(e) &= \sum_{e\in \delta(L)} w(e) + 2\sum_{e\in \delta(L_1,L_2)} w(e) \\
\geq \sum_{e\in \delta(L)} w(e) + m_L &\geq \frac{1}{2} \sum_{v \in L} W_G(v) \geq \frac{\alpha \mone}{2}
\end{align}
Hence, the better of the two solutions $T \setminus L_1$ or $T \setminus L_2$ is guaranteed to have a cut of weight at least $\frac{\alpha \mone}{4}$, where $\mone$ is the total number of \textsf{1-edges} in $G$. To complete the proof we note that for any optimal solution $OPT$, $w(\delta(OPT)) \leq \mone$.
\qed
\end{proof}

Thus, if we have an algorithm to compute $\alpha$-thick trees, Lemma \ref{lem:alphatt} provides an $\Omega(\alpha)$-approximation algorithm for the \bcmc\ problem. Unfortunately, there exist graphs that do not contain $\alpha$-thick trees for any non-trivial value of $\alpha$. For example, let $G$ be a \emph{path graph} with $n$ vertices and $m = n-1$ \textsf{1-edges}. It is easy to see that for any subtree $T$, the sum of degrees of the leaves is at most 4. In spite of this setback, we show that the notion of $\alpha$-thick trees is still useful in obtaining a good approximation algorithm for the \bcmc\ problem. In particular, Lemma \ref{lem:good-leaves} and Theorem \ref{thm:logn} show that path graph is the \emph{only} bad case, i.e., if the graph $G$ does not have a long induced path, then one can find an $\Omega(\frac{1}{\log n})$-thick tree. Lemma \ref{lem:deg2} shows that we can assume without loss of generality that the \bcmc\ instance does not have such a long induced path.
\\

\subsubsection*{Shrinking Thin Paths.}
A natural idea to handle the above ``bad'' case is to get rid of such long paths that contain only vertices of degree two by contracting the edges. We refer to a path that only contains vertices of degree two as a \textsf{d-2} path.
Further, we define the length of a \textsf{d-2} path as the number of \emph{vertices} (of degree two) that it contains. The following lemma shows that we can assume without loss of generality that the graph $G$ contains no ``long'' \textsf{d-2} paths.

\begin{restatable}{lem}{degtwo}
Given a graph $G$, we can construct, in polynomial time, a graph $G'$ with no \textsf{d-2} paths of length $\geq 3$ such that $G'$ has a \bcmc\ solution $S'$ of cut weight ($w(\delta(S'))$) at least $\psi$ if and only if $G$ has a \bcmc\ solution $S$ of cut weight at least $\psi$. Further, given the solution $S'$ of $G'$, we can recover $S$ in polynomial time.
\label{lem:deg2}
\end{restatable}

\begin{proof}
We may assume that $G$ is connected, because otherwise we can handle each component separately. We further assume that $G$ is not a simple cycle, otherwise it is trivial to solve such an instance. If $G$ does not have a \textsf{d-2} path of length $\geq 3$, then trivially we have $G' = G$. Otherwise, let $\wp = [v_0,e_0,v_1,e_1,v_2,e_2,v_3]$  be a path in $G$ such that $v_1, v_2$ and $v_3$ have degree two and $deg(v_0) \neq 2$. Note that such a path $\wp$ must exist as $G$ is not a simple cycle. 
We now perform the following operation on $G$ to obtain a new graph $G_{new}$: Delete these elements $\{e_0, v_1, e_1, v_2, e_2\}$. Add a new vertex $v_{new}$ and edges $e_0' =(v_0,v_{new})$ and $e_1' = (v_{new}, v_3)$. Since $deg(v_0) \neq 2$ and $deg(v_3) = 2$, we are guaranteed that $v_0 \neq v_3$ and hence we do not introduce any multi-edges. 
The weights on the new edges are determined as follows - Let $n_\wp$ denote the number of \textsf{1-edges} in $E_\wp = \{e_0, e_1,e_2\}$. If $n_\wp \geq 2$, we set $w(e'_0) =w(e'_1) = 1$. If $n_\wp = 1$, then we set $w(e'_0) = 0$ and $w(e'_1) = 1$. Otherwise, we set $w(e'_0) = w(e'_1) = 0$. We claim that $G_{new}$ has a \bcmc\ solution $S'$ of cut weight at least $\psi$ if and only if $G$ has a solution $S$ of cut weight at least $\psi$. 

Let us first assume that there is a set $S$ in $G$ that is a solution to the \bcmc\ problem with cut weight $\psi$. We now show that there exists a $S'$ in $G_{new}$ that is a solution to the \bcmc\ problem with cut weight at least $\psi$. The proof in this direction is done for three possible cases, based on the cardinality of $\delta(S)\cap E_\wp$. We note that $|\delta(S)\cap E_\wp|$ is $\leq 2$, since $G[S]$ must be connected.

{\bf Case 1.} $|\delta_G(S)\cap E_\wp| = 2$. Note that since $S$ is connected, we must have either (i) $S \subseteq \{v_1,v_2\}$ or (ii) $\{v_0,v_3\} \subseteq S$. In the former case, we set $S' =  \{v_{new}\}$ and the claim follows by the definition of $w(e'_0)$ and $w(e'_1)$. In the latter case, we set $S' = S \setminus \{v_1, v_2\}$. Since $v_1$ and $v_2$ are vertices of degree two, $G_{new}[S']$ is connected. Further, every edge $e \in \delta_G(S) \setminus E_\wp$ also belongs to $\delta_{G_{new}}(S')$. The claim follows once we observe that both $e'_0$ and $e'_1$ are in $\delta_{G_{new}}(S')$. 

{\bf Case 2.} $|\delta_G(S)\cap E_\wp| = 1$. In this case, we must have either $v_0 \in S$ or $v_3 \in S$ but not both. Let us first assume $v_0\in S$. We set $S' = (S\cup \{v_{new}\})\setminus \{v_1, v_2\}$. It is clear that if $G[S]$ is connected, so is $G_{new}[S']$. 
Due to the removal of $v_1$ and $v_2$, we have $\delta_G(S) \setminus \delta_{G_{new}}(S') = \{e_i\}$ for some edge $e_i \in E_\wp$. On the other hand, due to the addition of $v_{new}$, we have $\delta_{G_{new}}(S') \setminus \delta_{G}(S) = \{e'_1\}$ and the claim follows since $w(e'_1) \geq w(e_i)$ for any $e_i \in E_\wp$.
Now assume that $v_3\in S$. In this case, we set $S' = S\setminus \{v_1, v_2\}$. Since $v_{new}\notin S'$, we again have $e_1' \in \delta_{G_{new}}(S')$ and the proof follows as above. 

{\bf Case 3.} $|\delta_G(S)\cap E_\wp| = 0$. In this case, one of the following holds, either (i) $\{v_0,v_1,v_2,v_3\} \subseteq S$ or (ii) $\{v_0,v_1,v_2,v_3\} \cap S = \phi$. If the latter is true, the proof is trivial by setting $S' = S$. In the former case, we set $S' = S \setminus \{v_1,v_2\} \cup \{v_{new}\}$. The addition of $v_{new}$ maintains connectivity between $v_0$ and $v_3$ and hence since $S$ is connected, so is $S'$. Further, we have $\delta_G(S) = \delta_{G_{new}}(S')$ since no edge in $\delta_G(S)$ in incident on $v_1$ or $v_2$.

  In order to prove the other direction, we assume that $S'$ is a solution to the \bcmc\ problem on $G_{new}$ with a cut weight of $\psi$. We now construct a set $S$ that is a solution to \bcmc\ on $G$ of weight at least $\psi$. The proof proceeds in three cases similarly. 

{\bf Case 1.} Both $e_0' \in \delta_{G_{new}}(S')$ and $e'_1 \in \delta_{G_{new}}(S')$. One of the following holds - (i) $S' = \{v_{new}\}$ or (ii) $\{v_0, v_3\} \subseteq S'$.
In the former case, let $S$ be the subset of $\{v_1, v_2\}$ having the largest weight cut. By construction, we have that weight of the cut $\delta(S)$ is at least the sum of weights of $e_0'$ and $e_1'$. For the latter, let $S$ to be the best among $S', S' \cup \{v_1\},$ and $S' \cup \{v_2\}$ and the proof follows as above.

{\bf Case 2.} Either $e_0'\in \delta_{G_{new}}(S')$ or $e_1'\in \delta_{G_{new}}(S')$ but not both. Let $e_{max}$ be the edge of maximum weight in $E_\wp$. The edge $e_{max}$ splits the path $\wp$ into two connected components one containing $v_0$, call it $\wp_0$ and the other containing $v_3$, call it $\wp_3$. Now to construct $S$, we delete $v_{new}$ from $S'$ (if it contains it) and add the component $\wp_0$ if $v_0\in S'$ or the component $\wp_3$ if $v_3\in S'$. Again connectivity is clearly preserved. We now argue that the cut weight is also preserved. Indeed, this is true since we have that $w(e_{max}) \geq max(w_{e_0'}, w_{e_1'})$ and the rest of the cut edges in $S'$ remain as they are in $S$.

 {\bf Case 3.} None of $e_0', e_1'$ belong to $\delta_{G_{new}}(S')$. In this case, if $v_{new}\notin S'$, then trivially $S = S'$ works. Otherwise, we set $S = S' \cup \{v_1, v_2\}$. It is easy to observe that both connectivity and all the cut edges are preserved in this case. 

Now, to construct $G'$, we repeatedly apply the above contraction as long as possible. This will clearly take polynomial time as in each iteration, we reduce the number of degree-2 vertices by 1. 
 Hence we have the claim.   
\qed
\end{proof}

\subsubsection*{Spanning Tree with Many Leaves.}
Assuming that the graph has no long \textsf{d-2} paths, the following lemma shows that we can find a spanning tree $T$ that has $\Omega(n)$ leaves. Note that Claim \ref{lem:only-good} now guarantees that there are $\Omega(n)$ \textsf{1-edges} incident on the leaves of $T$. 

\begin{restatable}{lem}{goodleaves}
Given a graph $G=(V,E)$ with no \textsf{d-2} paths of length $\geq 3$, we can obtain, in polynomial time, a spanning tree $T = (V,E_T)$ with at least $\frac{n}{14}$ leaves. 
\label{lem:good-leaves}
\end{restatable}
\begin{proof}
  Let $T$ be any spanning tree of $G$. We note that although $G$ does not have \textsf{d-2} paths of length $\geq$ 3, such a guarantee does not hold for paths in $T$. Suppose that there is a \textsf{d-2} path $\wp$ of length 7 in $T$. Let the vertices of this path be numbered $v_1, v_2, \ldots, v_7$ and consider the vertices $v_3, v_4, v_5$. Since $G$ does not have any \textsf{d-2} path of length 3, there is a vertex $v_i, i \in \{3,4,5\}$ such that $deg_G(v_i) \geq 3$. We now add an edge $e = \{v_i, w\}$ in $G \setminus T$ to the tree $T$. The cycle $C$ that is created as a result must contain either the edge $\{v_1,v_2\}$ or the edge $\{v_6,v_7\}$. We delete this edge to obtain a new spanning tree $T'$. It is easy to observe that the number of vertices of degree two in $T'$ is strictly less than that in $T$. This is because, although the new edge $\{v_i, w\}$ can cause $w$ to have degree two in $T'$, we are guaranteed that the vertex $v_i$ will have degree three and vertices $v_1$ and $v_2$ (or $v_6$ and $v_7$) will have degree one. Hence, as long as there are \textsf{d-2} paths of length 7 in $T$, the number of vertices of degree two can be strictly decreased.  Thus this process must terminate in at most $n$ steps and the final tree $T^{(1)}$ obtained does not have any \textsf{d-2} paths of length $\geq 7$.

We now show that the tree $T^{(1)}$ contains $\Omega(n)$ leaves by a simple charging argument. Let the tree $T^{(1)}$ be rooted at an arbitrary vertex. We assign each vertex of $T^{(1)}$ a token and redistribute them in the following way : Every vertex $v$ of degree two in $T^{(1)}$ gives its token to its first non degree two descendant, breaking ties arbitrarily. Since there is no \textsf{d-2} path of length $\geq 7$, each non degree two vertex collects at most 7 tokens. Hence, the number of vertices not having degree two in $T^{(1)}$ is at least $\frac{n}{7}$. Further, since the average degree of all vertices in a tree is at most 2, a simple averaging argument shows that $T^{(1)}$ must contain at least $\frac{n}{14}$ vertices of degree one, i.e., $\frac{n}{14}$ leaves.
\qed
\end{proof}

\subsection*{Obtaining an $\Omega(\frac{1}{\log n})$ Approximation}

We now have all the ingredients required to obtain the $\Omega(\frac{1}{\log n})$ approximation algorithm. We observe that if the graph $G$ is sparse, i.e. $\mone \leq c n \log n$ (for a suitable constant $c$), then the tree obtained by using Lemma \ref{lem:good-leaves} is an $\Omega(\frac{1}{\log n})$-thick tree and thus we obtain the required approximate solution in this case. On the other hand, if the graph $G$ is sparse, then we use Lemma \ref{lem:good-leaves} to obtain a spanning tree, delete the leaves of this tree, and then repeat this procedure until we have no more vertices left. 
Since, we delete a constant fraction of vertices in each iteration, the total number of iterations is $O(\log n)$. We then choose the ``best'' tree out of the $O(\log n)$ trees so obtained and show that it must be an $\alpha$-thick tree, with $\alpha=\Omega(\frac{1}{\log n})$. 
Finally, using Lemma \ref{lem:alphatt}, we obtain an $\Omega(\frac{1}{\log n})$ approximate solution as desired. We refer to Algorithm \ref{alg:logn} for the detailed algorithm.

\begin{algorithm}[htbp]
\DontPrintSemicolon
 {\bf Input}: Graph $G = (V, E)$\;
 {\bf Output:} A subset $S\subseteq V$, such that $G[S]$ is connected\;
Set $G_1(V_1,E_1) = G$, $n_1 = |V_1|$\;
Let $\mone \leftarrow$ Number of \textsf{1-edges} in $G$\;
Use Lemma~\ref{lem:good-leaves} to obtain a spanning tree $T_1$ of $G_1$ with leaf set $L_1$\;
\If{$\mone \leq c n \log n$} {
  Use Lemma~\ref{lem:alphatt} on $T_1$ to obtain a set connected $S$\;
  \Return $S$\;
}
$i = 1$\;
 \While{$G_i \neq \phi$}{\label{algm:line:alphatt}
    $E_{i+1} \leftarrow E_i \setminus (E[L_i] \cup \delta(L_i))$ \label{algm:line:delete_edges} \;
    $V_{i+1} \leftarrow V_i \setminus L_i$, $n_{i+1} = |V_{i+1}|$\;
    Contract degree-2 vertices in $G_{i+1}$ \label{algm:line:contract_edges}\;
    Use Lemma~\ref{lem:good-leaves} to obtain a spanning tree $T_{i+1}$ of $G_{i+1}$ with leaf set $L_{i+1}$\;
    $i = i+1$\;
 }
Choose $j = \argmax_i(\sum_{v \in L_i} deg_G(v))$\;
Use Lemma~\ref{lem:alphatt} on $T_j$ to obtain a connected set $S$\;
\Return $S$\;
\caption{Finding \alphatt\ trees}
\label{alg:logn}
\end{algorithm}

\begin{restatable}{thm}{main}
  \label{thm:logn}
  Algorithm \ref{alg:logn} gives an $\Omega(\frac{1}{\log n})$ approximate solution for the \bcmc\ problem.
\end{restatable}
\begin{proof}
  Let us assume that $\mone \leq c n \log n$ (for some constant $c$). Now,  Lemma \ref{lem:good-leaves} and Claim \ref{lem:only-good} together imply that $\sum_{v \in L_1} W_G(v) = \Omega(n)$.
Further, since we have $w(\delta(OPT)) \leq \mone \leq c n \log n$, $T$ is an $\alpha$-thick tree for some $\alpha = \Omega(\frac{1}{\log n})$. Hence, we obtain an $\Omega(\frac{1}{\log n})$ approximate solution using Lemma \ref{lem:alphatt}.

On the other hand, if $\mone > c n \log n$, we show that at least one of the trees $T_i$ obtained by the repeated applications of the Lemma \ref{lem:good-leaves} is an $\alpha$-thick tree $T$ of $G$ for $\alpha = \Omega(\frac{1}{\log n})$. 
 We first observe that the \textsf{While} loop in Step \ref{algm:line:alphatt} runs for at most $O(\log n)$ iterations. This is because we delete $\Omega(n_i)$ leaves in each iteration and hence after $k = O(\log n)$ iterations, we get $G_k = \phi$.
We now count the number of \textsf{1-edges} ``lost'' in each iteration. We recall that $W_G(v)$ is the total number of \textsf{1-edges} incident on $v$ in a graph $G$. In an iteration $i$, the number of \textsf{1-edges} lost at Step \ref{algm:line:delete_edges} is at most $\sum_{v \in L_i} W_{G_{i}}(v)$. In addition, we may lose a total of at most $2n \leq \frac{2\mone}{c \log n}$ edges due to the contraction of degree two vertices in Step \ref{algm:line:contract_edges}. 
Suppose for the sake of contradiction that $\sum_{v \in L_i} W_{G}(v) < \frac{\mone}{d \log n}, \forall 1 \leq i \leq k$ where $d$ is a suitable constant. Then the total number of \textsf{1-edges} lost in $k = O(\log n)$ iterations is at most
\[
\sum_{i = 1}^k (\sum_{v \in L_i} W_{G_{i}}(v)) + \frac{2\mone}{c \log n}
< \sum_{i = 1}^k \frac{\mone}{d \log n} + \frac{2\mone}{c \log n} = \frac{\mone}{\hat{d}} + \frac{\mone}{c \log n} < \mone \]
The equality follows for a suitable constant $\hat{d}$ as $k = O(\log n)$. The final inequality holds for a suitable choice of the constants $c$ and $d$.
But this is a contradiction since we have $G_k = \phi$.

Since we choose $j$ to be the best iteration, we have $\sum_{v \in L_j}W_{G}(v) \geq \frac{\mone}{d \log n}$ for some constant $d$. Hence the tree $T_j$ is an $\alpha$-thick tree of $G$ for $\alpha = \frac{1}{d\log n}$ and the theorem follows by Lemma \ref{lem:alphatt}.
\qed
\end{proof}

\subsection*{General Weighted Graphs}
\label{sec:weighted-graphs}
We now consider the \textsf{Weighted Connected Maximum Cut} (\wcmc) problem. Formally, we are given a graph $G=(V,E)$ and a weight function $w:E\rightarrow \mathbb{R}^+\cup\{0\}$. The goal is to find a subset $S$ of vertices that induces a connected subgraph and maximizes the quantity $\sum_{e\in \delta(S)} w(e)$.
 We obtain a $\Omega(\frac{1}{\log^2n})$ approximation algorithm for this problem. 
Our basic strategy is to group edges having nearly the same weight into a class and thus create $O(\log n)$ classes. We then solve the \bcmc\ problem for each class independently and return the best solution.

\begin{algorithm}[htbp]
 {\bf Input}: Connected graph $G = (V, E)$ with $|V| = n$ and $|E| = m$; Weight function, $w:E\rightarrow \mathbb{R}^+\cup{0}$ ; $\epsilon > 0 $\;
 {\bf Output:} A subset $S\subseteq V$, such that $G[S]$ is connected\;
Let $w_{max}$ be the maximum weight over any edge of the graph\;
Define, $w_0 = \frac{\epsilon w_{max}}{m}$ and $w_i = w_0(1+\epsilon)^i$, for $i\in [\log_{1+\epsilon} \frac{m}{\epsilon}]$\;
\For{$i\in [0, \log_{1+\epsilon} \frac{m}{\epsilon}]$}
{
\For{$e \in E$}{
\If{$ w_i \leq w(e) < w_{i+1}$}{
$w'_i(e) = 1$\;
}
\Else{
$w'_i(e)=0$\;
}
Using Theorem~\ref{thm:logn}, solve for the connected subset $S_i$\;
}
}
\Return $S_{best}$, such that  $best = \displaystyle \argmax_{i\in[0,\log \frac{n}{\epsilon}]} \sum_{e\in \delta(S_i)}w(e)$\;
\caption{Algorithm for the \textsf{Weighted} \cmcfull\ problem.}
\label{alg:log2n}

\end{algorithm}

\begin{theorem}
Algorithm~\ref{alg:log2n} gives a $\Omega(\frac{1}{\log^2n})$ approximation guarantee for the \textsf{Weighted} \cmcfull\ problem. 
\end{theorem}
\begin{proof}
Let $OPT$ be an optimal solution for a given instance of the problem and let $\psi = \sum_{e\in \delta(OPT)} w(e)$. Also, let $\epsilon\in  (0,1]$.
Since we have that $\psi \geq w_{max}$, we can reset the weights of those edges with weight $< \frac{\epsilon w_{max}}{m}$ to 0 and  assume that $w_{min} \geq \frac{\epsilon w_{max}}{m}$ where $w_{min}$ denotes the weight of the minimum (non zero) weight edge.
Let $E_i$ be the set of edges $e$ such that $w_i\leq w(e) < w_{i+1}$ and finally let $OPT_i = \delta(OPT)\cap E_i$. 
We now claim that $\sum_{e\in OPT_i} w(e) = O((1+\epsilon)\log n\sum_{e\in \delta(S_i)}w(e))$. This immediately gives us that $\sum_{e\in \delta(S_{best})}w(e) = \Omega(\frac{\sum_{e\in OPT}w(e)}{(1+\epsilon)\log n\log_{1+\epsilon}\frac{m}{\epsilon}}) = \Omega({\frac{1}{\log^2n})(\sum_{e\in OPT}w(e)})$. 

We now prove the claim. Consider solving the \bcmc\ instance with weight function $w_i'$. Clearly $OPT$ is a feasible solution to this instance and we have $\sum_{e\in \delta(OPT)}w'_i(e) = \sum_{e\in OPT_i} w'_i(e) \leq O(\log n \sum_{e\in \delta(S_i)}w'_i(e))$. The previous inequality holds as $S_i$ is guaranteed to be an $\Omega(\frac{1}{\log n})$-approximate solution by Theorem \ref{thm:logn}. Now, we have $\sum_{e\in OPT_i} w(e) \leq (1+\epsilon)w_i\sum_{e\in OPT_i} w'_i(e) \leq  O((1+\epsilon)w_i \log n\sum_{e\in \delta(S_i)} w'_i(e)) \leq O((1+\epsilon)\log n\sum_{e\in \delta(S_i)} w(e))$. Hence, the claim.    
\qed
\end{proof}

\section{CMC in Planar and Bounded Genus Graphs}
In this section, we consider the \cmc\ problem in planar graphs and more generally, in graphs with genus bounded by a constant. We show that the \cmc\ problem has a PTAS in bounded genus graphs.

\subsection*{PTAS for Bounded Genus Graphs.}
\label{sec:bounded-genus-graphs}
We use the following (paraphrased) contraction decomposition theorem by Demaine, Hajiaghayi and Kawarabayashi~\cite{demaine2011contraction}.

\begin{restatable}{thm}{demaine}\emph{(\cite{demaine2011contraction})}\label{thm:demaine} For a bounded-genus graph $G$ and an integer $k$, the edges of $G$ can be partitioned into $k$ color classes such that contracting all the edges in any color class leads to a graph with treewidth $O(k)$. Further, the color classes are obtained by a radial coloring and have the following property: 
If edge $e=(u,v)$ is in class $i$, then every edge $e'$ such that $e' \cap e \neq \phi$ is in class $i-1$ or $i$ or $i+1$.
\end{restatable}

Given a graph $G$ of constant genus, we use Theorem \ref{thm:demaine} appropriately to obtain a graph $H$ with constant treewidth. In Appendix \ref{sec:dptw}, we show that one can solve the \cmc\ problem optimally in polynomial time on graphs with constant treewidth.

\begin{restatable}{thm}{ptas}
If the \cmc\ problem can be solved optimally on graphs of constant treewidth, then there exists a polynomial time $(1-\epsilon)$ approximation algorithm for the \cmc\ problem on bounded genus graphs (and hence on planar graphs).  
\end{restatable}

\begin{proof}
Let $G = (V,E)$ be the graph of genus bounded by a constant and let $\opt$ denote the optimal \cmc\ of $G$ and $\psi = |\delta(\opt)|$ be its size.
Using Theorem \ref{thm:demaine} with $k = \frac{3}{\epsilon}$, we obtain a partition of the edges $E$ into $\frac{3}{\epsilon}$ color classes namely $C_1, C_2, \ldots, C_{\frac{3}{\epsilon}}$. We further group three consecutive color classes into $\frac{1}{\epsilon}$ groups $G_1, \ldots, G_{\frac{1}{\epsilon}}$ where $G_j = C_{3j-2} \cup C_{3j-1} \cup C_{3j}$. Let $G_{j^*}$ denote the group that intersects the least with the optimal connected max cut of $G$, i.e., $j^* = \argmin_j(|G_j \cap \delta(\opt)|)$\footnote{We ``guess'' $j^*$ by trying out all the $\frac{1}{\epsilon}$ possibilities}. As the $\frac{1}{\epsilon}$ groups partition the edges, we have $|G_{j^*} \cap \delta(\opt)| \leq \epsilon \psi$. Let $i = 3j^* - 1$, so that $G_{j^*} = C_{i-1} \cup C_i \cup C_{i+1}$. Let $H = (V_H, E_H)$ denote the graph of treewidth $O(\frac{1}{\epsilon})$ obtained by contracting all edges of color $C_i$.

We first show that $H$ has a \cmc\ of size at least $(1-\epsilon) \psi$. For a vertex $v \in V_H$, let $\mu(v) \subseteq V$ denote the set of vertices of $G$ that have merged together to form $v$ due to the contraction. We define a subset $S' \subset V_H$ as  $S' = \{v \in V_H \ |\ \mu(v) \cap \opt \neq \phi\}$.
Note that because we contract edges (and not delete them), $S'$ remains connected. We claim that $|\delta(S')| \geq (1 - \epsilon) \psi$. 
Let $e = (u,v)$ be an edge in $\delta(\opt)$. Now $e \notin \delta(S')$ implies that at least one edge $e'$ such that $e' \cap e \neq \phi$ has been contracted. By the property guaranteed by Theorem \ref{thm:demaine}, we have that $e \in G_{j^*}$. Hence we have, $|\delta(S')| \geq |\delta(S) \setminus G_{j^*}| = |\delta(S)| - |G_{j^*} \cap \delta(S)| \geq (1-\epsilon) \psi$.

Finally, given a connected max cut of size $\psi$ in $H$, we can recover a connected max cut of size at least $\psi$ in $G$ by simply un-contracting all the contracted edges. Hence, by solving the \cmc\ problem on $H$ optimally, we obtain a $(1-\epsilon)$ approximate solution in $G$.
\qed
  \end{proof}
  
\subsection*{NP-hardness in planar graphs}
\label{sec:np-hardness}
We now describe a non-trivial polynomial time reduction of a \textsf{3-SAT} variant known as \textsf{Planar Monotone 3-SAT} (\pmsat) to the \cmc\ problem on a planar graph, thereby proving that the latter is NP-hard. The following reduction is interesting as the classical \mc\ problem can be solved optimally in polynomial time on planar graphs using duality. In fact, it was earlier claimed that  even \cmc\ can be solved similarly~\cite{haglin1991approximation}.

An instance of \pmsat\ is a 3-CNF boolean formula $\phi$ such that -
\vspace{-3mm}
\begin{enumerate}
\item[a)] A clause contains either all positive literals or all negative literals.
\item[b)] The associated bipartite graph $G_{\phi}$\footnote{$G_{\phi}$ has a vertex for each clause and each variable and an edge between a clause and the variables that it contains} is planar.
\item[c)] Furthermore, $G_\phi$ has monotone, rectilinear representation. We refer the reader to Berg and Khosravi~\cite{de2008finding} for a complete description.
Figure~\ref{rectilinear} illustrates the rectilinear representation by a simple example.  
\end{enumerate}

\begin{figure}[htb]
        \centering
        \begin{subfigure}[b]{0.45\textwidth}
                \centering
                \includegraphics[width=\textwidth]{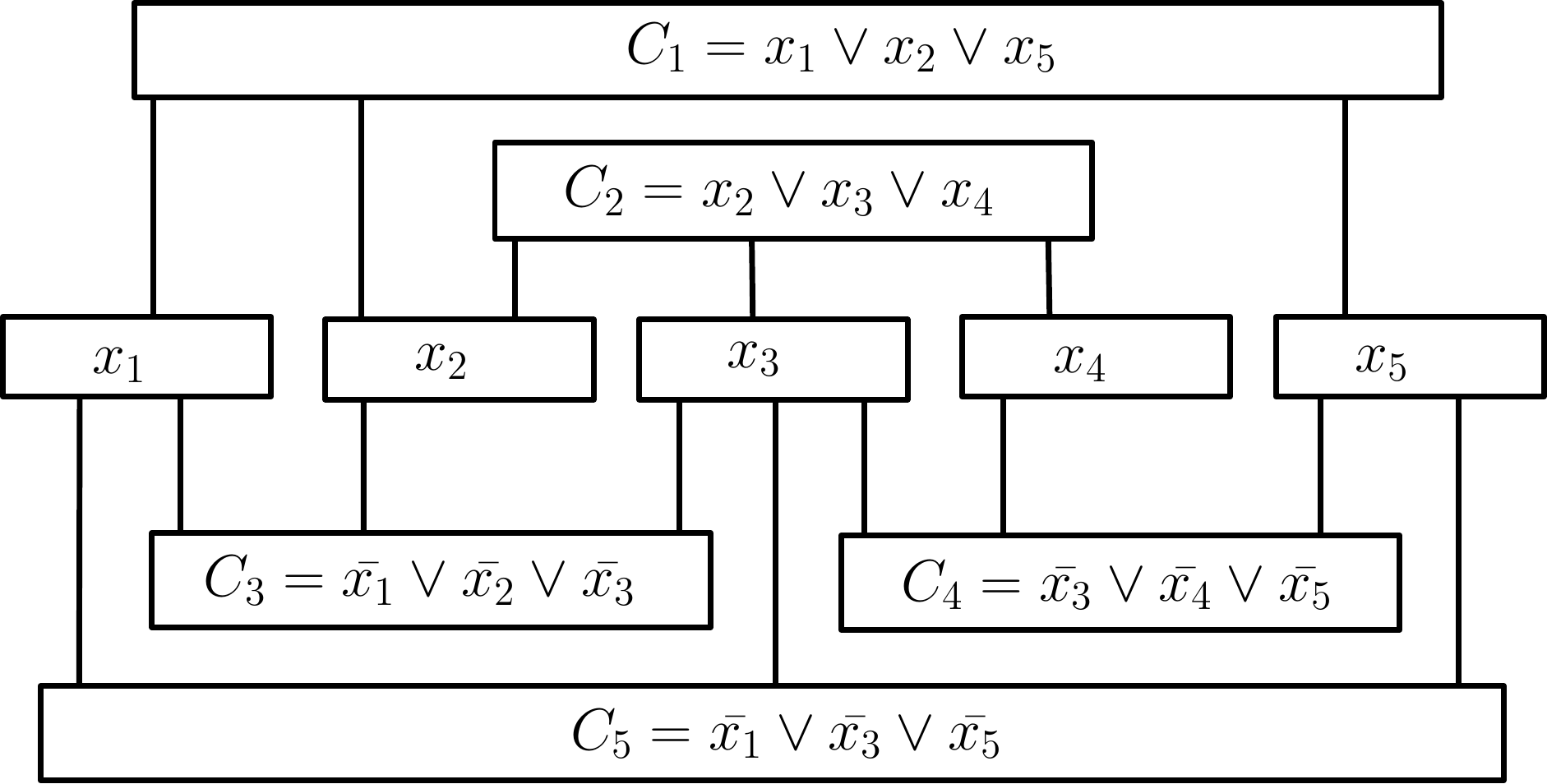}
                \caption{Monotone Rectilinear Representation}
                \label{rectilinear}
         \end{subfigure}
	\quad\quad
        \begin{subfigure}[b]{0.45\textwidth}
                \centering
                \includegraphics[width=\textwidth]{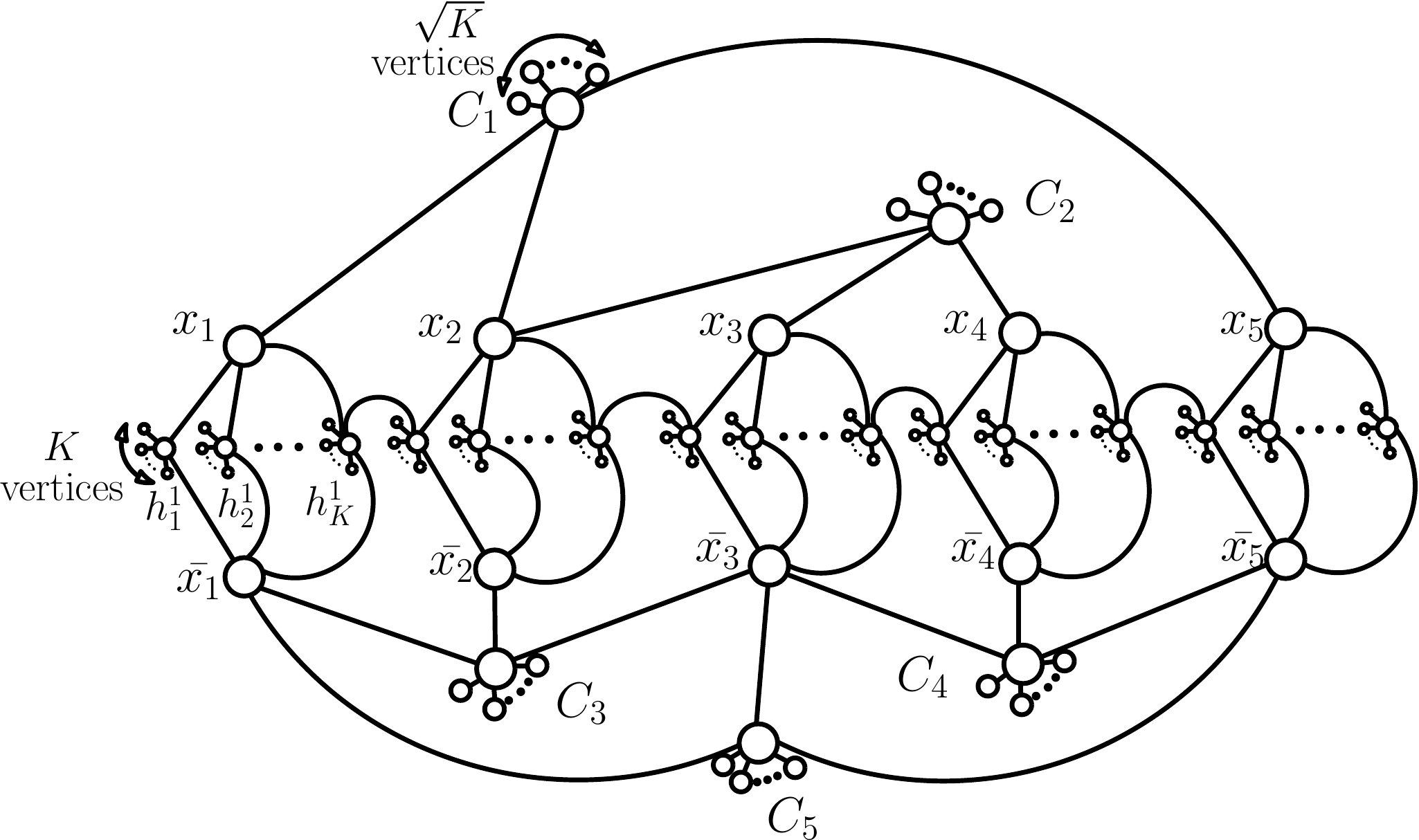}
                \caption{Reduction of PM-3SAT to a \textsf{Planar } \cmc\ instance}
                \label{reduction}
        \end{subfigure}
        \caption{Example illustrating the rectilinear representation and the reduction to a \textsf{Planar }\cmc\ instance of the formula $(x_1\vee x_2 \vee x_5) \wedge (x_2\vee x_3 \vee x_4) \wedge (\bar{x_1}\vee \bar{x_2}\vee \bar{x_3})\wedge (\bar{x_3}\vee \bar{x_4}\vee \bar{x_5}) \wedge (\bar{x_1}\vee \bar{x_3}\vee \bar{x_5}). $}\label{figall}
\end{figure}

Given such an instance, the \pmsat\ problem is to decide whether the boolean formula is satisfiable or not. Berg and Khosravi~\cite{de2008finding} show that the \pmsat\ problem is NP-complete.\\

\noindent {\bf The Reduction.}  Given a \pmsat\ formula $\phi$, with a rectilinear representation, we obtain a polynomial time reduction to a \textsf{Planar }\cmc\ instance, there by showing that the latter is NP-hard. 
Let $\{x_i\}_{i=1}^n$ denote the variables of the \pmsat\ instance and $\{C_j\}_{j=1}^m$ denote the clauses. We construct a planar graph $H_{\phi}$ as follows. For every variable $x_i$, we construct the following gadget:  We create two vertices $v(x_i)$ and $v(\bar{x_i})$ corresponding to the literals $x_i$ and $\bar{x_i}$.
   Additionally, we have $K > m^2 $  ``helper" vertices, $h^i_1, h^i_2, \ldots, h^i_K$ such that each $h^i_k$ is adjacent to both $x_i$ and $\bar{x_i}$.
    Further, for every $h^i_k$ we add a set $L^i_k$ of $K$ new vertices that are adjacent only to $h^i_k$.
    Now, in the rectilinear representation of the \pmsat, we replace each variable rectangle by the above gadget. 
    For two adjacent variable rectangles in the rectilinear representation, say $x_i$ and $x_{i+1}$, we connect the helpers $h^i_K$ and $h^{i+1}_1$. 
    For every clause $C_j$, $H_{\phi}$ has a corresponding vertex $v(C_j)$ with edges to the three literals in the clause. 
    Finally, for each vertex $v(C_j)$, we add a set $L_j$ of $\sqrt{K}$ new vertices adjacent only to $v(C_j)$.
It is easy to observe that the reduction maintains the planarity of the graph. Figure~\ref{reduction} illustrates the reduction by an example. 

We show the following theorem that proves the \textsf{Planar} \cmcfull\ problem is NP-hard. 

\begin{restatable}{thm}{nphard}
  Let $H_{\phi}$ denote an instance of the planar \cmc\ problem corresponding to an instance $\phi$ of \pmsat\ obtained as per the reduction above. Then, the formula $\phi$ is satisfiable if and only if there is a solution $S$ to the \cmc\ problem on $H_{\phi}$ with $|\delta_{H_\phi}(S)| \geq m\sqrt{K} + nK+ nK^2$.
\end{restatable}

\begin{proof}
For brevity, we denote $\delta_{H_\phi}(S)$ as $\delta(S)$ in the rest of the proof. 

{\bf Forward direction.}
Assume that $\phi$ is satisfiable under an assignment $A$. We now show that we can construct a set $S$ with the required properties.
Let $\{l_i\}_{i\in [n]}$ be the set of literals that are true in $A$. 
We define $S = \{v(l_i)\}_{i \in [n]} \cup \{C_j\}_{j \in [m]} \cup \{h^i_k\}_{i \in [n], k \in [K]}$, i.e., the set of vertices corresponding to the true literals, all the clauses and all the helper vertices. By construction, the set of all helper vertices and one literal of each variable induces a connected subgraph. Further, since in a satisfying assignment every clause has at least one true literal, the constructed set $S$ is connected. We now show that $|\delta(S)|\geq m\sqrt{K} + nK +nK^2$.
Indeed, $\delta(S)$ contains all the edges corresponding to the one degree vertices incident on clauses and all the helpers. This contributes a profit of $m\sqrt{K} + nK^2$. Also, since no vertex corresponding to a false literal is included in $S$ but all helpers are in $S$, we get an additional profit of $K$ for each variable. Hence, we have the claim. 

{\bf Reverse direction.} Assume that $S$ is a subset of vertices in $H_\phi$ such that $H_\phi[S]$ is connected and $|\delta(S)| \geq m\sqrt{K} + nK+ nK^2$. We now show that $\phi$ is satisfiable. We may assume that $S$ is an optimal solution (since optimal solution will satisfy these properties, if a sub-optimal solution does). We first observe that at least one of the (two) literals for each variable must be chosen into $S$. Indeed, if this is not the case for some variable, for $H_\phi[S]$ to be remain connected, none of the helper vertices corresponding to that variable can be chosen. This implies that the maximum possible value for $|\delta(S)| \leq (n-1)K^2 + m\sqrt{K} + 3m + 2(n-1)K$ (this is the number of remaining edges)$< nK^2$ (since $K> m^2$) $ < m\sqrt{K} + nK + nK^2$, a contradiction. We now show that every helper vertex must be included in $S$. Assume that this is not true and let $h^i_k$ be some helper vertex not added to $S$. We note that none of the $K$ degree one vertices in $L^i_k$
 can be in $S$ because $H_\phi[S]$ must be connected. Now, consider the solution $S'$ formed by adding $h^i_k$ to $S$. Since at least one vertices $v(x_i)$ or $v(\bar{x_i})$ is in $S$, if $H_\phi[S]$ is connected, so is $H_\phi[S']$. Further, the total number of edges in the cut increases by $K-2$. This is a contradiction to the fact that $S$ is an optimal solution. Hence, every helper vertex $h^i_k$ belongs to the solution $S$. We now show that, no two literals of the same variable are chosen into $S$. Assume the contrary and let $v(x_i)$, $v(\bar{x_i})$ both be chosen into $S$. We claim that removing one of these two literals will strictly improve the solution. Indeed, consider removing $v(x_i)$ from $S$. Clearly, we gain all the edges from $v(x_i)$ to all the helper vertices corresponding to this variable. Thus we gain at least $K$ edges. We now bound the loss incurred. In the worst case, removing $v(x_i)$ from $S$ might force the removal of all the clause vertices due to the connectivity restriction. But this would lead to a loss of at most $m\sqrt{K} + 3m < K$. Hence, we arrive at a contradiction that $S$ is an optimal solution. Therefore, exactly one literal vertex corresponding to each variable is included in $S$. Finally, we observe that all the clauses must be included in $S$. Assume this is not true and that $m'< m$ clause vertices are in $S$. Now the total cut is $nK + nK^2 + m'\sqrt{K} < nK+nK^2+m\sqrt{K}$, which is again a contradiction. Now, the optimal solution $S$ gives a natural assignment to the \pmsat\ instance: a literal is set to \textsf{TRUE} if its corresponding vertex is included in $S$. Since, every clause vertex belongs to $S$, which in turn is connected, it must contain a \textsf{TRUE} literal and hence the assignment satisfies $\phi$. 
\qed
 \end{proof}

\bibliographystyle{plain}
\bibliography{bibfile}

\appendix
\section{Dynamic program for constant tree-width graphs}
\label{sec:dptw}

The notion of \textsf{tree decomposition} and \textsf{tree-width} was first introduced by Robertson and Seymour~\cite{robertson1984graph}. Given a graph $G=(V,E)$, its tree decomposition is a tree representation $T = (\mathcal{B}, \mathcal{E})$, where each $b\in \mathcal{B}$ (called as a \emph{bag}) is associated with a subset $B_b\subseteq V$ such that the following properties hold:
\begin{enumerate}
\item $\bigcup_{b\in B} B_b = V$. 
\item For every edge $u,v\in E$, there is a bag $b\in \mathcal{B}$, such that $u,v\in B_b$.
\item For every $u\in V$, the subgraph $T_u$ of $T$, induced by bags that contain $u$, is connected. 
\end{enumerate}
The width of a decomposition is defined as the size of the largest bag $b \in B$ minus one.
Treewidth of a graph is the minimum width over all the possible tree decompositions. 
In this section, we show that the \cmc\ problem can be solved optimally in polynomial time on graphs with constant treewidth $t$. 

\noindent {\bf Notation.} We denote the tree decomposition of a graph $G=(V,E)$ by $\mathcal{T} = (\mathcal{B}, \mathcal{E})$. For a given bag of the decomposition $b\in \mathcal{B}$, let $B_b$ denote the set of vertices of $G$ contained in $b$ and $V_b$ denote the set of vertices in the subtree of $\mathcal{T}$ rooted at $b$.
 As shown by Kloks~\cite{kloks1994treewidth}, we may assume that $\mathcal{T}$ is \textsf{nice tree decomposition}, that has the following additional properties.
\begin{enumerate}
\item Any node of the tree has at most $2$ children.
\item A node $b$ with no children is called a \textsf{leaf} node and has $|B_b| = 1$.
\item A node $b$ with two children $c_1$ and $c_2$ is called a \textsf{join} node. For such a node, we have $B_b = B_{c_1} = B_{c_2}$.
\item A node $b$ with exactly one child $c$  is either a \textsf{forget} node or an \textsf{introduce} node. If $b$ is a \textsf{forget} node
then $B_b = B_c\setminus \{v\}$ for some $v\in B_c$. One the other hand, if $b$ is an \textsf{introduce} node then $B_b = B_c\cup \{v\}$ for $v\notin B_c$.
\end{enumerate}

We now describe a dynamic program to obtain the optimal solution for the \cmc\ problem. Let $OPT$ denote the optimal solution. We first prove the following simple claim that helps us define the dynamic program variable.\\

\begin{claim}
For any bag $b\in \mathcal{B}$, the number of components induced by $OPT\cap V_b$ in $G$ is at most $t$.
\end{claim} 

\begin{proof}
Consider the induced subgraph $G[OPT\cap V_b]$ and let $C$ be one of its components. We observe that $C$ has at least one vertex in $B_b$, i.e., $C\cap B_b \neq \phi$. Assume this is not true and $C\cap B_b = \phi$. Now consider an edge $e =(u,v)$ such that $u\in C$ and $v\in OPT\setminus C$. Such an edge is guaranteed to exist owing to the connectivity of $G[OPT]$. By our assumptions, $v\notin V_b$. This implies there is some bag $b'$ not in the subtree of $\mathcal{T}$ rooted at $b$, that contains both $u$ and $v$. But this in turn implies $u\in B_b$, a contradiction to the assumption $C\cap B_b = \phi$. Now, since each vertex in $B_b$ belongs to at most one component, there can be at most $|B_b|\leq t$ components in $G[OPT\cap V_b]$. Hence, the claim. 
\end{proof}

\noindent For a given $b\in \mathcal{B}$, let $S_b = OPT\cap B_b$ be the set of vertices chosen by the optimal solution from  the bag $b$. Further, let $P_b = (C_1, C_2\ldots C_t)$ be a  partition, of size $t$, of the vertices in $S_b$, such that each non-empty $C_i$ (some of the $C_i$'s could possibly be empty) is a subset of a unique component of the subgraph induced by $V_b\cap OPT$. We now define the variable of the dynamic program $M_b(P_b,S_b)$ in the following way: Consider the subgraph induced by $V_b$ in $G$ and let $S$ be a subset of $V_b$ with maximum cut $\delta_{G[V_b]}(S)$, such that every $C_i\in P_b$  is completely contained in a distinct component of $G[S]$. We set $M_b(P_b, S_b) = |\delta_{G[V_b]}(S)|$.

 From this definition, it follows that the optimal solution can be obtained by computing $M_r(P_r = (S,\phi,\phi, \ldots, \phi), S)$, for every subset $S$ of $B_r$, where $r$ is the root bag of the tree $\mathcal{T}$ and picking the best possible solution. We now describe the dynamic program to compute the above variable $M_b(P_b, S_b)$ for a given bag $b$. 
 
\noindent {\bf Case 1:} \textsf{Node $b\in \mathcal{T}$ is a leaf node}. In this case, $B_b= V_b = \{v\}$, for some vertex $v$.
\begin{align}
M_b((\{v\}, \phi, \ldots \phi), \{v\}) &= 0 \nonumber \\
M_b((\phi, \phi, \ldots \phi),\phi) &= 0 \nonumber
\end{align}

\noindent {\bf Case 2:} \textsf{Node $b\in \mathcal{T}$ is an introduce node}. In this case, $b$ has exactly one child node $c$  and $B_b = B_c\cup \{v\}$, for some vertex $v$.
 Let $P_b = (C_1,C_2\ldots C_t)$ be some partition of $S_b\subseteq B_b$. 
We compute $M_b(P_b, S_b)$ as follows. 
\begin{align*}
\intertext{If $v\notin S_b$, then}
M_b(P_b,S_b) &= M_c(P_b, S_b) + |\delta_{G[B_b]}(v, S_b)| 
\intertext{If  $v\in C_i$ and $v$  is  adjacent  to some vertex in  $C_i\setminus \{v_i\}$ but not to any vertex in $C_j, j\neq i$, then}
M_b(P_b,S_b) &= M_c((C_1,C_2, \ldots, C_i\setminus \{v\}, \ldots, C_t),S_b\setminus\{v\}) + |\delta_{G[B_b]}(v,B_b\setminus S_b)| 
\intertext {In all other cases, we set}
M_b(P_b,S_b) &= -\infty
\end{align*}
We now argue about the correctness of this case. First assume that $v\notin S_b$, that is $v$ is not chosen into our solution. If $S \subseteq V_c$ is the set of vertices chosen into our solution so far, the total cut size increases by $|\delta_{G[V_c]}(v,S)|$, which we claim is equal to $|\delta_{G[B_b]}(v,S_b)|$. In other words, none of the edges incident on $v$ are adjacent to any vertices in $S\setminus S_b$. Suppose for the sake of contradiction that there exists a vertex $w\in S\setminus S_b$ be such that $\{v,w\}\in E$. This implies that there exists a bag $b'$, not in the subtree of $b$, that contains both $v$ and $w$. This in turn implies  that $w\in S_b$, which is a contradiction to the fact that $w\in S\setminus S_b$. Hence, in this case the total cut increases by $|\delta_{G[B_b]}(v,S_b)|$. 

Now assume that $v\in S_b$, more specifically, let $v\in C_i$, for some $i$. From the above argument there are no edges between $v$ and vertices in $S\setminus S_b$. Since  we must have all  vertices in $C_i$ in a single component of $G[V_b]$ and all edges incident on $v$ in $G[V_b]$ have the other end in $B_b$,  $v$ must have an edge to some vertex in $C_i\setminus \{v\}$. Further, since any $C_i$ and $C_j$ must belong to distinct components, they must not share any vertices. Thus, if $v$ either has no edges to $C_i\setminus \{v\}$ or has an edge to some $C_j$, there is no feasible solution and we assign $-\infty$ to this variable. On the other hand if both these conditions are satisfied, our solution is valid and the increase in the cut size is $|\delta_{G[V_b]}(v, B_b\setminus S_b)|$.    

\noindent{\bf Case 3:} \textsf{Node $b\in \mathcal{T}$ is a forget node}. In this case, again, $b$ has exactly one child node $c$ with $B_b= B_c \setminus \{v\}$.
Let $P_b = (C_1,C_2\ldots C_t)$. It is easy to see that:
\[
M_b(P_b, S_b) = 
\max \left\{ 
\begin{array} {ll}
M_c(P_b, S_b) &\\
M_c((C_1,C_2, \ldots, C_i\cup \{v\}, \ldots, C_t), S_b\cup\{v\}), &\text{}\forall i\in[t] 
\end{array} 
\right.
\]

\noindent {\bf Case 4:}\textsf{ Node $b\in \mathcal{T}$ is a join node}. In this case, $b$ has two children $c_1, c_2$ with $B_b = B_{c_1} = B_{c_2}$. 
Let $P_b = (C_1, C_2\ldots, C_t)$  be the partition of vertices in $S_b$ such that each $C_i$ belongs to a distinct component in $G[OPT\cap V_b]$. 
Similarly, define $P_{c_1} = (C_1^1, C_2^1 \ldots C_t^1)$ and $P_{c_2} = (C_1^2, C_2^2 \ldots C_t^2)$ as partitions of $S_{c_1} = S_b$ and $S_{c_2} = S_b$ respectively. Consider the construction of following auxiliary graph, that we refer to as ``merge graph'', denoted by $M$. 
 For every $C_i^1$ and $C_i^2$, we have a corresponding vertex $v(C_i^1)$ and $v(C_i^2)$ respectively in $M$. 
 Further if two sets $C_i^1$ and $C_j^2$ intersect, i.e., $C_i^1\cap C_j^2 \neq \phi$, then we add an edge between $v(C_i^1),v(C_j^2)$ in $M$. 
 It is easy to observe that for a given component $C$ of $M$, the union of all subsets of vertices corresponding to the vertices of $C$ must belong to the same component of $G[OPT\cap V_b]$. Thus in turn implies that there is a one to one correspondence between $C_i\in P_b$ and the components of $M$. 
 
For a given partition $P_b$ of $S_b$, we call two partitions $P_{c_1}$ and $P_{c_2}$ of $S_{c_1}$ and $S_{c_2}$ as ``valid'' if there is a one to one correspondence between $C_i\in P_b$ and the components of $M$, as described above. We now prove the following simple claim.\\ 

\begin{claim}
For any  $S\subseteq V_b$,
$\delta_{G[V_b]}(S, V_b\setminus S) = \delta_{G[V_{c_1}]}(S_1, V_{c_1}\setminus S_1) + \delta_{G[V_{c_2}]}(S_2, V_{c_2}\setminus S_2) - \delta_{G[V_b]}(S_b,B_b\setminus S_b)$, where $S_1 = S\cap V_{c_1}$ and $S_2 = S\cap V_{c_2}$
\end{claim}
\begin{proof}

From the properties of tree decomposition, it follows that for any two vertices $u\in V_{c_1}\setminus B_{b}$ and $w\in V_{c_2}\setminus B_b$, $uw\notin E$. 
Further all the edges in $\delta(V_{c_1}\setminus B_{b})$ and $\delta(V_{c_1}\setminus B_{b})$ are incident on the vertices in $B_b$. Thus we have the following equation,
\begin{align}
\delta_{G[V_b]}(S,V_b\setminus S) =& \delta_{G[V_{c_1}]}(S_1\setminus B_b, V_{c_1}\setminus S_1) + \delta_{G[V_{c_2}]}(S_2\setminus B_b, V_{c_2}\setminus S_2)+ \delta_{G[V_b]}(S_b,B_b\setminus S_b) \nonumber\\
=&(\delta_{G[V_{c_1}]}(S_1\setminus B_b, V_{c_1}\setminus S_1) + \delta_{G[V_{b}]}(S_b,B_b\setminus S_b))\nonumber \\ &+ (\delta_{G[V_{c_1}]}(S_2\setminus B_b, V_{c_2}\setminus S_2) + \delta_{G[V_{b}]}(S_b,B_b\setminus S_b)) -  \delta_{G[V_{b}]}(S_b,B_b\setminus S_b) \nonumber\\
=&  \delta_{G[V_{c_1}]}(S_1, V_{c_1}\setminus S_1) + \delta_{G[V_{c_2}]}(S_2, V_{c_2}\setminus S_2) - \delta_{G[V_{b}]}(S_b,B_b\setminus S_b) \nonumber
\end{align}

\end{proof}
We can now compute the dynamic program variable, in this case, as follows. 
\[
M_b(P_b, S_b) = 
\displaystyle \max_{\substack{\text{$(P_{c_1},P_{c_2})$ are valid}\\ \textsf{with respect to $P_b$}}} M_{c_1}(P_{c_1}, S_b) + M_{c_2}(P_{c_2}, S_b) - |\delta_{G[V_b]}(S_b,B_b\setminus S_b)|
\]

\section{Maximum Leaf Spanning Tree}
\label{app:mlst}
In the Maximum Leaf Spanning Tree  problem, we are given an undirected graph $G=(V,E)$ and the goal is to find a spanning tree that has a maximum number of leaves. We now show that this well studied problem is a special case of the \textsf{Connected Submodular Maximization} problem. Consider maximizing the submodular function $f(S) = |\{v\ |\ v \in \mathbb{N}(S) \setminus S\}|$ where $\mathbb{N}(S)$ is the set of vertices that have a neighbor in $S$. Now, it is easy to observe that there is a tree with $L$ leaves if and only if there is a connected set $S$ such that $f(S) = L$. Further, we show that without loss of generality for any solution $S$ to Connected Submodular Maximization problem, we have that $S \cup \mathbb{N}(S) = V$ and hence the corresponding tree is spanning. Suppose $S$ is a feasible solution and $V \neq S \cup \mathbb{N}(S)$, then there must exist an edge $(u,v)$ such that $u \in \mathbb{N}(S)\setminus S$ and $v \notin S \cup \mathbb{N}(S)$. Now $S' = S \cup \{u\}$ is also a feasible solution and $f(S) = f(S')$.

\section{High Connectivity}
\label{app:conn}
We observe that the \textsf{Connected Submodular Maximization} has a constant factor approximation algorithm if the graph has high connectivity. 
Let $S \subseteq V$ be a random set of vertices such that every vertex $v$ is chosen in $S$ independently with probability $\frac{1}{2}$.
As shown by Feige et al.~\cite{feige2011maximizing}, $\mathbb{E}[f(S)] \geq (\frac{1}{4}) f(OPT)$ where the $f(OPT)$ is the maximum value of $f$. Further, as shown by Censor-Hillel et al.~\cite{censor2015tight} if the graph $G$ has $\Omega(\log n)$ vertex connectivity, the set $S$ obtained above is connected with high probability. Hence, $S$ is a $\frac{1}{4}$ approximate solution to the \textsf{Connected Submodular Maximization} problem.

\end{document}